\documentclass[submission,copyright,creativecommons]{eptcs}
\providecommand{\event}{ICE 2017} %
\usepackage{underscore}           %

\usepackage{ucs}
\usepackage[utf8x]{inputenc}

\usepackage{amsmath}
\usepackage{amssymb}
\usepackage{subcaption}

\usepackage[normalem]{ulem}
\usepackage{paralist}

\usepackage[dvipsnames]{xcolor}

\usepackage{tikz}

\RequirePackage{tikz}
\RequirePackage{xspace}
\RequirePackage{etoolbox}

\ifdef{\comma}{\renewcommand\comma{,\xspace}}{\newcommand\comma{,\xspace}}

\usetikzlibrary{arrows,positioning,petri,decorations,shapes}
\tikzset{,
  node distance=10mm,
  every place/.style=
    {
      circle,
      draw,
      thick,
      inner sep=3pt,
      minimum size=7mm
    },
  every transition/.style=
    {
      rectangle,
      draw,
      thick,
      inner sep=3pt,
      minimum size=7mm
    },
  edge/.style=
    {
      ->,
      shorten <=1pt,
      >=stealth',
      semithick
    },
  pre/.style=
    {
      <-,shorten <=1pt,>=stealth',semithick
    },
  post/.style=
    {
      ->,shorten >=1pt,>=stealth',semithick
    },
  state/.style=
    {
      circle,draw,semithick,inner sep=.1pt,minimum size=1.5mm,fill=black
    }
}
\RequirePackage{xspace}
\RequirePackage{amsmath,amsfonts}

\newcommand\xdot{.\xspace}

\newcommand\eg{e.\,g\xdot,\xspace}
\newcommand\ie{i.\,e\xdot,\xspace}
\newcommand\cf{cf\xdot}
\newcommand\etal{et al\xdot}
\newcommand\wrt{w.\,r.\,t\xdot}

\newcommand\traces[1]{\ensuremath{\mathsf{Traces}(#1)}\xspace}

\newcommand\teq{\ensuremath{\equiv_{T}^{\infty}}\xspace}

\newcommand\feq{\ensuremath{\equiv_{F}^{\infty}}\xspace}
\newcommand\req{\ensuremath{\equiv_{R}^{\infty}}\xspace}

\newcommand\pfeq{\ensuremath{\equiv_{PF}^{\infty}}\xspace}
\newcommand\seq{\ensuremath{\equiv_{S}^{\infty}}\xspace}

\newcommand\xlongrightarrow[1]{\ensuremath{\overset{#1}{\longrightarrow}}\xspace}

\newcommand\parl{\ensuremath{\,|\,}\xspace}

\newcommand\transrel[1]{\ensuremath{\xlongrightarrow{#1}}\xspace}

\newcommand\bisim{\ensuremath{\equiv_B}\xspace}

\newcommand\failtraceeq{\ensuremath{\equiv_{FT}^{\infty}}\xspace}

\newcommand\readytraceeq{\ensuremath{\equiv_{RT}^{\infty}}\xspace}

\newcommand{\fairAtt}{\ensuremath{\mathcal{F}}\xspace}
\newcommand{\wFairOp}[1]{\ensuremath{\varphi_{#1}}\xspace}
\newcommand{\sFairOp}[1]{\ensuremath{\Phi_{#1}}\xspace}
\newcommand\sfairop[1]{\sFairOp{#1}}
\newcommand\wfairop[1]{\wFairOp{#1}}
\newcommand{\fairatt}{\fairAtt}

\newcommand{\sFairLang}[1]{\ensuremath{\mathcal{L}^\Phi_{#1}}\xspace}
\newcommand{\wFairLang}[1]{\ensuremath{\mathcal{L}^\varphi_{#1}}\xspace}

\newcommand\parll{\ensuremath{\,\|\,}\xspace}

\newcommand\runs{\ensuremath{\textit{Rn}}\xspace}

\newcommand\trace{\ensuremath{\textit{tr}}\xspace}

\usepackage{amsthm}
\theoremstyle{plain}
\newtheorem{theorem}{Theorem}

\theoremstyle{definition}
\newtheorem{definition}[theorem]{Definition}

\ifx\numberwithinsect\relax
  \@addtoreset{theorem}{section}
  \edef\thetheorem{\expandafter\noexpand\thesection\@thmcountersep\@thmcounter{theorem}}
\fi

\def\titlerunning{Keep it Fair: Equivalences}
\title{\titlerunning
\footnote{This work was funded by the DFG (German Research Foundation), grants GO-671/6-2 and NE-1505/2-2.}
}
\author{Tobias Prehn 
\institute{Modelle und Theorie Verteilter Systeme\\ TU Berlin, Germany}
\email{tobias.prehn@tu-berlin.de}
\and Stephan Mennicke
\institute{Institut für Programmierung und Reaktive Systeme\\ TU Braunschweig, Germany}
\email{mennicke@ips.cs.tu-bs.de}}
\def\authorrunning{T.\ Prehn \& S.\ Mennicke}

\begin{document}
\maketitle

\begin{abstract}
For models of concurrent and distributed systems, it is important and also challenging to establish correctness in terms of safety and/or liveness properties.
Theories of distributed systems consider equivalences fundamental, since they (1) preserve desirable correctness characteristics and (2) often allow for component substitution making compositional reasoning feasible.
Modeling distributed systems often requires abstraction utilizing nondeterminism which induces unintended behaviors in terms of infinite executions with one nondeterministic choice being recurrently resolved, each time neglecting a single alternative.
These situations are considered unrealistic or highly improbable.
Fairness assumptions are commonly used to filter system behaviors, thereby distinguishing between realistic and unrealistic executions.
This allows for key arguments in correctness proofs of distributed systems, which would not be possible otherwise.
Our contribution is an equivalence spectrum in which fairness assumptions are preserved.
The identified equivalences allow for (compositional) reasoning about correctness incorporating fairness assumptions.
\end{abstract}

\section{Introduction} %
\label{sec:introduction}
In theories of concurrent and distributed systems, nondeterminism is central for describing system behavior.
Thereby, we highly abstract from implementation details in the decision-making processes to be realized, while focusing on all possible system behaviors.
A common abstraction therefore is nondeterminism, especially when it comes to the description of alternative behaviors.
Nondeterminism in system modeling falls apart into two categories.
{\em Internal nondeterminism} are those situations in which the observation of an action does not allow for the prediction of the successor state.
On the other hand, {\em external nondeterminism} are circumstances where the successor state may be determined, but it is impossible to predict or influence which one of two or more conflicting actions occurs, \eg the choice for one specific action may be controlled by an environment the system works in.
External nondeterminism is ubiquitous in distributed systems where usually two or more components contribute to the system behavior being modeled by an interleaving of the component actions.
For liveness properties (``something good eventually happens''), nondeterminism gives rise to system behaviors considered unrealistic.
The ``something good'' may be in reach infinitely many times but each time, the corresponding action may be evaded.
In this paper, we study the implications nondeterminism has on proofs for liveness properties.

In 1965, Edsger W.\ Dijkstra unleashed a controversial debate on system correctness proofs with his seminal paper on a problem in concurrent programming, today coined to the term {\em mutual exclusion}~\cite{Dijkstra65}.
Roughly a year later, Donald E.\ Knuth noticed and reported an error in Dijkstra's algorithm by showing an erroneous execution sequence, contradicting one aspect of the algorithm, and in conclusion, the proof of correctness as given by Dijkstra~\cite{Knuth66}.
The execution sequence Knuth took constitutes a corner case where one of two agents must be assumed to be indefinitely faster than the other component in order to constantly block the access to the critical section. 
Let us recapitulate the situation with a less complex example, given by the system depicted in Fig.~\ref{fig:intro-example}(a).
The system behavior of a process $q$ is modeled in a {\em labeled transition system} (LTS) containing $q$.
The use of LTS gives us a well studied, commonly used formalism and a wide applicability of our results since there exist numerous encodings from other formalisms into LTS and vice versa.
When considering the behavior of a process, the respective state may be seen as the initial state of the LTS.
One of two things may happen:
\begin{enumerate}
\item Either action $a$ occurs, leading to a state where action $c$ is momentarily disabled
\item or action $c$ occurs and afterwards action $a$ may occur infinitely often.
\end{enumerate}
For the sake of illustration, assume that action $a$ is performed by a component $K_1$ and the observation of action $c$ is due to component $K_2$.
The intuition is that $K_1$ uses some resource that $K_2$ needs to perform action $c$ and returns it every even occurrence of action $a$, as long as $c$ did not occur.
Assume that a system correctness property requires us to prove that action $c$ eventually occurs, \eg in order to ensure that $K_2$ eventually terminates.
Of course, we directly derive a counterexample by an infinite sequence of actions $a$, \ie $w=a^{\omega}$ disproves the desired theorem.

In the example, whenever a recurring decision is independent of outer circumstances, one option may continuously be selected in favor of another, leading to an unfair treatment of the neglected actions, and surely also of the neglected component.
This problem is often addressed by posing fairness assumptions, guaranteeing the elimination of the aforementioned unfair treatment of components and/or actions.
At this level of abstraction, it is simply unrealistic that one component, here $C_2$, is infinitely slower than another, when competing for a limited resource.
It is sufficient to assume that $C_2$ eventually gets a grip on the limited resource to perform action $c$, as long as an even number of actions $a$ occurs infinitely often.
With this assumption in mind, it is now possible to rule out the counterexample $w$.
All other infinite behaviors certainly contain action $c$ thus ensuring the required theorem.
Thus, a fairness assumption defines a certain balance between the theoretically possible and realistic system behavior.
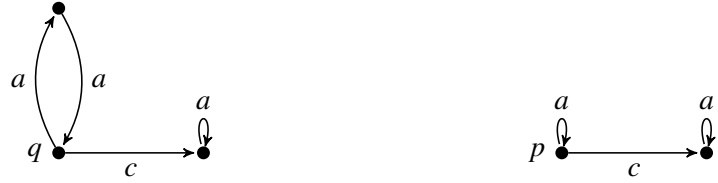
\begin{figure}[t]
\centering
\begin{subfigure}[b]{.4\textwidth}
\centering
\begin{tikzpicture}
  \node[state,label=left:$q$] (q0) {};
  \node[state] at (0,5em) {}
    edge[pre,bend right=30] node[auto,swap] {$a$} (q0)
    edge[post,bend left=30] node[auto] {$a$} (q0);
  \node[state] (q1) at (5em,0) {}
    edge[pre] node[auto] {$c$} (q0);
  \draw[edge,loop above] (q1) to node[auto] {$a$} (q1);
\end{tikzpicture}
\end{subfigure}\quad
\begin{subfigure}[b]{.4\textwidth}
\centering
\begin{tikzpicture}
  \node[state,label=left:$p$] (p0) {};
  \draw[edge,loop above] (p0) to node[auto] {$a$} (p0);
  \node[state] (p1) at (5em,0) {}
    edge[pre] node[auto] {$c$} (p0);
  \draw[edge,loop above] (p1) to node[auto] {$a$} (p1);
\end{tikzpicture}
\end{subfigure}
\caption{Two processes for which (a) strong fairness of $c$ and (b) weak fairness of $c$ ensures that eventually only actions $a$ are enabled.}\label{fig:intro-example}
\end{figure}

In practice, process creation involves specification and implementation, allowing for several design decisions of the same effect.
To identify systems that are capable of the same behavior and to match implementations against specifications, several equivalence notions are formulated, each meeting different requirements \wrt correctness, \eg deadlock-freedom.
Some of those equivalences form a hierarchy~\cite{Glabbeek1990Linear,Glabbeek1993} and many process logics, as \eg HML, LTL, or CTL, enjoy the property of {\em bisimulation invariance}~\cite{Hennessy85, Moller99}.
The latter constitutes an important insight in process theory, stating that whenever two systems are equivalent under bisimulation, then they enjoy the same properties as formulated in a certain logic.
Regardless of the strengths of the equivalences in van Glabbeek's spectrum, most of them do not consider system assumptions, such as fairness, stated outside of the used formalism.
However, in a system development process, when refining a system specification, such as the one given in Fig.~\ref{fig:intro-example}(a), to another system that is not bisimilar to the specification, it is unclear to what degree also the assumptions of fairness are carried over.
For example, trace equivalence is a candidate that changes the set of fair runs from one process to another equivalent process, being sketched in Sect.~\ref{sec:strong-fairness}.
This means that, even if we need to reprove certain properties for trace equivalent system refinements, the same fairness assumptions as posed upon the specification yields new paths to be taken into account.
We prove for all equivalences at least as strong as failure trace equivalence, that equivalent systems yield the same fair behavior.
On the other hand, equivalences at most as strong as ready equivalence do not preserve fairness in this respect.

We consider strong fairness as well as weak fairness, as introduced by Plotkin~\cite{Plotkin1982} or Francez~\cite{Francez1986}.
Formal definitions are adapted from notions defined by Reisig~\cite{Reisig1998}.
A system run is strong fair \wrt some action iff either this action occurs infinitely often or it is enabled only finitely often.
An action $a$ is {\em enabled} in a process if it is possible to execute $a$, thereby reaching another process.
In our example of Fig.~\ref{fig:intro-example}(a), assuming that action $c$, and thereby component $C_2$, is treated strong fair, ensures the success of the theorem.
Consider the same theorem for the system in Fig.~\ref{fig:intro-example}(b).
Here, also strong fairness for action $c$ is sufficient for proving the desired property.
In this case, strong fairness is not necessary, since also weak fairness suffices.
A system run is weak fair \wrt some action, here $c$, iff either it occurs infinitely often or it is not always enabled from some state within the run.
Since in Fig.~\ref{fig:intro-example}(b), process $p$ enables action $c$ and leaves it enabled after an action $a$ occurs, the assumption of weak fairness ensures that $c$ needs to be disabled in a weak fair run, again ruling out the canonical counterexample $w=a^{\omega}$.

The paper is structured as follows.
Sect.~\ref{sec:preliminaries} introduces the notion of labeled transition systems and a selection of semantic equivalences relevant to this work.
The notions of fairness are defined in Sect.~\ref{sec:strong-fairness} and Sect.~\ref{sec:weak-fairness}, for each of which an area within the linear-time branching-time spectrum is identified for which fairness is preserved.
We conclude the paper by a discussion on related work and plans for future work by Sect.~\ref{sec:related-work}.
\section{Preliminaries} %
\label{sec:preliminaries}
Here, we cover the definitions of labeled transition systems (LTS) as well as a selection of equivalences on LTS related and classified by van Glabbeek in \cite{Glabbeek1990Linear}.
LTS is a formalism to describe abstract behaviors while LTS equivalences relate such behaviors to one another.
An LTS is a state-transition graph where each transition is labeled by a letter from an alphabet $\Sigma$, the set of all (abstract) actions.
\begin{definition}[Labeled Transition System]\label{def:lts}
  Let $\Sigma$ be an alphabet.
  A {\em labeled transition system} (LTS) is a triple $A = ( Q, \Sigma, \transrel{})$ with a set of {\em processes} $Q$ and a labeled transition relation $\transrel{}\,\subseteq Q\times \Sigma\times Q$.
\end{definition}
Throughout the paper, processes range over $p, q, r, s, p', q'$ as well as $p_i$ and $q_i$ for $i \in \mathbb{N}$.
We abbreviate $(p, a, p')\in\,\transrel{}$ to $p\transrel{a} p'$.
$p\transrel{a}$ denotes enabledness of action $a$ in process $p$, \ie there exists $p' \in Q$ such that $p\transrel{a} p'$.
For a transition $t=(p,a,p')\in\,\transrel{}$, $l(t) := a$ defines its label projection.
Fairness notions are based on \emph{runs} of processes.
A run of a process $p_0$ is an alternating sequence of processes and actions, initiated by $p_0$.
Runs come in finite or infinite length.
Traces of infinite runs are elements of $\Sigma^{\omega}$, \ie all infinite words over $\Sigma$.
\begin{definition}[Process Run]\label{def:lts-runs}
Let $A = (Q, \Sigma, \transrel{})$ be an LTS.
A {\em finite run of $p_0 \in Q$} is a sequence of the form $\rho = p_0 a_1 \ldots p_{n-1} a_n p_n$ with $p_i\in Q$ and $a_i\in\Sigma$ such that $p_{i}\transrel{a_{i+1}} p_{i+1}$ ($0\leq i<n$).
An {\em infinite run of $p_0 \in Q$} is an infinite sequence of the form $\rho = p_0 a_1 p_1 a_2 \ldots$ with $p_i\in Q$ and $a_i\in\Sigma$ such that $p_{i}\transrel{a_{i+1}} p_{i+1}$ ($0 \leq i$).
  The set of all finite and infinite runs of $p_0$ is denoted by $\runs(p_0)$.
\end{definition}
Subsequently, we introduce eight different semantics on processes, formally.
Therefore, we assume a single LTS $A=(Q,\Sigma,\transrel{})$ containing all processes mentioned throughout the following paragraphs.
All definitions carry over to a setting where each process is identified by a single LTS with distinct initial state.
As fairness assumptions deal with infinite runs only, we use the infinitary versions of trace-based semantics.
Trace semantics simply enumerates all finite and infinite action sequences of a process.
Two processes are equivalent if their sets of action sequences are equal.
\begin{definition}[Trace Semantics~\cite{Glabbeek1990Linear}]\label{def:traceeq}
  Let $A = (Q, \Sigma, \transrel{})$ be an LTS and $\rho = p_0 a_1 \ldots p_{n-1} a_n p_n$ be a run of $p_0$.
  The {\em trace of $\rho$} is defined by $\trace(\rho) := a_1 a_2 \ldots a_n\in\Sigma^{*}$.
  The transition relation of LTS naturally extends to traces $\sigma=a_1 a_2 \ldots a_n\in\Sigma^{*}$ as $p_0 \xlongrightarrow{\sigma} p_n$.
  The set of all finite traces of $p_0$ is denoted by $\textit{Traces}(p_0)$.
  
  Let $\rho' = p_0 a_1 p_1 a_2 \ldots$ be an infinite run of $p_0$.
  The {\em trace of $\rho'$} is defined by $\trace(\rho') := a_1 a_2 \ldots\in\Sigma^{\omega}$.
  The set of all infinite traces is defined by $\textit{Traces}^{\infty}(p_0)$.
  Two processes $p_0, q_0 \in Q$ are {\em infinitary trace equivalent}, denoted by $p_0 \teq q_0$ iff $\textit{Traces}(p_0)=\textit{Traces}(q_0)$ and $\textit{Traces}^{\infty}(p_0)=\textit{Traces}^{\infty}(q_0)$.
\end{definition}
Trace semantics is a linear-time semantics not distinguishing between systems producing the same traces via intermediate processes capable of different actions.
Ready semantics respects the branching structure by combining traces of a process with all actions enabled by a process reached by this trace, the so-called ready set. Therefore, ready sets can be seen as system logs for a specific point in time, stating the actions that have been taken so far and the actions being possible at this point. 
\begin{definition}[Ready Semantics~\cite{Glabbeek1990Linear}]\label{def:readyeq}
  Let $A = (Q, \Sigma, \transrel{})$ be an LTS.
  A {\em ready pair of $p_0 \in Q$} is a pair $(\sigma,X)\in \Sigma^{*}\times 2^{\Sigma}$ such that there exists a $p\in Q$ with $p_0\xlongrightarrow{\sigma} p$ and $X = \{ a\in\Sigma \parl p\transrel{a} \}$.
  The set of all ready pairs of $p_0$ is denoted by $R(p_0)$.
  Two processes $p_0, q_0 \in Q$ are {\em infinitary ready equivalent}, denoted $p_0 \req q_0$ iff $R(p_0)=R(q_0)$ and $p_0 \equiv_{T}^{\infty} q_0$.
\end{definition}
Although the infinitary versions of ready semantics respects infinite runs to some degree, branching is only considered to a finite extent.
Ready trace semantics, on the other hand, removes this restriction by integrating ready sets directly into finite and infinite traces.
In comparison to ready pairs, the ready traces can be seen as system logs where not only the taken actions are represented but also the alternative, enabled actions that have been neglected. This is achieved by allowing for (but not enforcing) intermediate ready sets.
\begin{definition}[Ready Trace Semantics~\cite{Glabbeek1990Linear}]\label{def:readytraceeq}
  A word $\sigma_1 \sigma_2 \ldots \sigma_n \in(\Sigma\cup 2^{\Sigma})^{*}$ is a {\em ready trace of $p_0 \in Q$} iff 
  there are states $p_1, p_2, \ldots, p_n \in Q$ such that for each $ 0 \leq i < n$ either
  \begin{itemize}
  \item[(a)] $\sigma_{i+1} \in \Sigma$ and $p_i \transrel{\sigma_{i+1}} p_{i+1}$ or
  \item[(b)] $\sigma_{i+1} \in 2^{\Sigma}$ and $p_i = p_{i+1}$ and $\sigma_{i+1} = \{ a\in\Sigma \parl p_i\transrel{a} \}$.
  \end{itemize}
  The set of all ready traces of $p_0$ is denoted by $RT(p_0)$.
  An {\em infinite ready trace of $p_0 \in Q$} is a word $\sigma_1\sigma_2 \ldots \in (\Sigma\cup 2^{\Sigma})^{\omega}$ iff there are states $p_1, p_2, \ldots \in Q$ such that for each $ 0 \leq i$ either (a) or (b) holds. 
  By $RT^{\infty}(p_0)$ we denote the set of all infinite ready traces of $p_0$.
  Two processes $p_0, q_0 \in Q$ are {\em infinitary ready trace equivalent}, denoted $p_0\readytraceeq q_0$ iff $RT(p_0)=RT(q_0)$ and $RT^{\infty}(p_0)=RT^{\infty}(q_0)$.
\end{definition}
A natural counterpart of ready semantics is failures semantics.
Its central idea relies on {\em failure pairs}, each consisting of a finite trace and a (not necessarily maximal) set of actions refused by a process reachable via the trace.
Failures semantics can be seen as system logs similar to ready semantics. They also state the actions taken up to some point in time but close with a set of impossible actions.
\begin{definition}[Failures Semantics~\cite{Glabbeek1990Linear}]\label{def:faileq}
  Let $A = ( Q, \Sigma, \transrel{})$ be an LTS.
  A {\em failure pair of $p_0 \in Q$} is a pair $(\sigma,X)\in \Sigma^{*}\times 2^{\Sigma}$ such that there exists a $p\in Q$ with $p_0\xlongrightarrow{\sigma} p$ and $X \subseteq \{ a\in\Sigma \parl p\not\transrel{a} \}$.
  The set of all failure pairs of $p_0$ is denoted by $F(p_0)$.
  Two processes $p_0, q_0 \in Q$ are {\em infinitary failures equivalent}, denoted $p_0 \feq q_0$ iff $F(p_0)=F(q_0)$ and $p_0 \equiv_{T}^{\infty} q_0$.
\end{definition}

\begin{figure}[t]
\centering
\scalebox{1}{\begin{tikzpicture}[font=\small]

\node (sbisi) at (-8em,0) {$\bisim$};
\node (rtrcs) at (0,0) {$\equiv_{RT}^{\infty}$};
\node (sbisi2rtrcs) at (-4em,0) {$\subset$};
\draw[semithick] (sbisi) to (sbisi2rtrcs) to (rtrcs);

\node (fail) at (8em,0) {$\equiv_F^{\infty}$};
\node (trcs) at (16em,0) {$\equiv_T^{\infty}$};
\node (fail2trcs) at (12em,0) {$\subset$};

\draw[semithick] (fail) to (fail2trcs) to (trcs);

\node (ftrcs) at (4em,-4em) {$\equiv_{FT}^{\infty}$};
\node (ready) at (4em,4em) {$\equiv_{R}^{\infty}$};

\node[rotate=-45] (fail2ready) at (6em,2em) {$\subset$};
\draw[semithick] (fail) to (fail2ready) to (ready);

\node[rotate=45] (ready2rtrcs) at (2em,2em) {$\subset$};
\draw[semithick] (ready) to (ready2rtrcs) to (rtrcs);

\node[rotate=-45] (ftrcs2rtrcs) at (2em,-2em) {$\subset$};
\draw[semithick] (ftrcs) to (ftrcs2rtrcs) to (rtrcs);

\node[rotate=45] (fail2ftrcs) at (6em,-2em) {$\subset$};
\draw[semithick] (fail) to (fail2ftrcs) to (ftrcs);

\node (pf) at (0em,8em) {$\equiv_{PF}^{\infty}$};
\node[rotate=-45] (pf2ready) at (2em,6em) {$\subset$};
\draw[semithick] (pf) to (pf2ready) to (ready);
\node[rotate=45] (sbisi2pf) at (-4em,4em) {$\subset$};
\draw[semithick] (sbisi) to (sbisi2pf) to (pf);

\node (s) at (8em,-8em) {$\equiv_{S}^{\infty}$};
\node[rotate=45] (s2trcs) at (12em,-4em) {$\subset$};
\draw[semithick] (s) to (s2trcs) to (trcs);
\node[rotate=-33.69] (sbisi2s) at (0em,-4em) {$\subset$};
\draw[semithick] (sbisi) to (sbisi2s) to (s);

\draw[thick] (14em,9em) node[below left, align=right] {\scshape{\scriptsize Branching}\\ \scshape{\scriptsize Time}} -- (14em,-9em) node[above right, align=left] {\scshape{\scriptsize Linear} \\ \scshape{\scriptsize Time}}; %

\draw[thick,dashed] (-2em,9em) node[below left, align=right] {\scshape{Strong \& Weak}\\ \scshape{Fairness Preserving}} -- (-2em,4.2em) -- (6em,-3.4em) -- (6em,-9em); %

\end{tikzpicture}}
\caption{How the semantics relate to each other, adapted from \cite{Glabbeek1990Linear}.}
\label{fig:semeq-hier}
\end{figure}
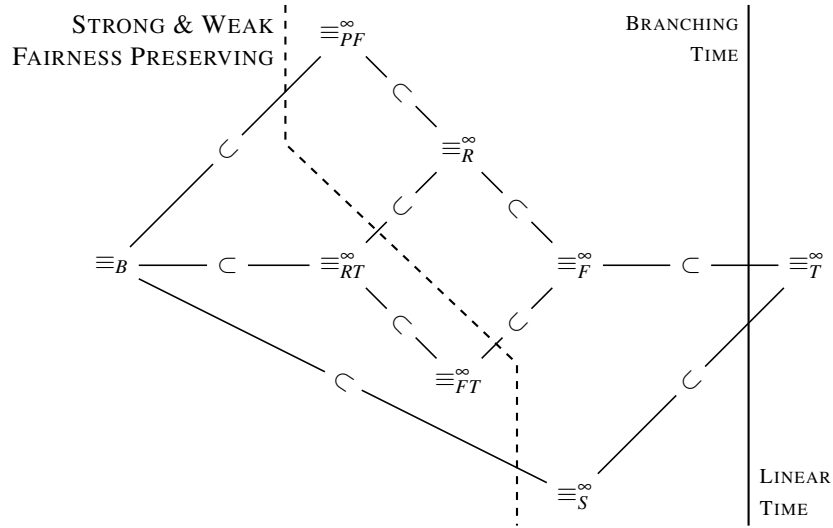
In contrast to the sets of all enabled actions that are integrated into the ready traces, the sets integrated into failure traces contain some of the refused actions.
\begin{definition}[Failure Trace Semantics~\cite{Glabbeek1990Linear}]\label{def:failtraceeq}
  A  word $\sigma_1 \sigma_2 \ldots \sigma_n \in(\Sigma\cup 2^{\Sigma})^{*}$ is a {\em failure trace of $p_0 \in Q$} iff there are states $p_1, p_2, \ldots, p_n \in Q$ such that for each $ 0 \leq i < n$ either
  \begin{itemize}
  \item[(a)] $\sigma_{i+1} \in \Sigma$ and $p_i \transrel{\sigma_{i+1}} p_{i+1}$ or
  \item[(b)] $\sigma_{i+1} \in 2^{\Sigma}$ and $p_i = p_{i+1}$ and $\sigma_{i+1} \subseteq \{ a\in\Sigma \parl p_i\not\transrel{a} \}$.
  \end{itemize}
  The set of all failure traces of $p_0$ is denoted by $FT(p_0)$.
  An {\em infinite failure trace of $p_0$} is a word $\sigma_1\sigma_2 \ldots \in (\Sigma\cup 2^{\Sigma})^{\omega}$ iff there are states $p_1, p_2, \ldots \in Q$ such that for each $ 0 \leq i$ either (a) or (b) holds. 
  By $FT^{\infty}(p_0)$ we denote the set of all infinite failure traces of $p_0$.
  Two processes $p_0, q_0 \in Q$ are {\em infinitary failure trace equivalent}, denoted $p_0\failtraceeq q_0$ iff $FT(p_0)=FT(q_0)$ and $FT^{\infty}(p_0)=FT^{\infty}(q_0)$.
\end{definition}
While ready semantics takes into account the set of enabled actions and, by this, the possible next step of the reached process, possible futures semantics considers all possible next sequences of actions \ie traces the reached process is capable of.
\begin{definition}[Possible Futures Semantics~\cite{Glabbeek1990Linear}]\label{def:possfuteq}
  Let $A = (Q, \Sigma, \transrel{})$ be an LTS.
  A {\em possible future of $p_0 \in Q$} is a pair $(\sigma,X)\in \Sigma^{*}\times 2^{\Sigma^{*}}$ such that there exists a $p\in Q$ with $p_0\xlongrightarrow{\sigma} p$ and $X = \traces{p}$.
  The set of all possible futures of $p_0$ is denoted by $PF(p_0)$.
  Two processes $p_0, q_0 \in Q$ are {\em infinitary possible futures equivalent}, denoted $p_0 \pfeq q_0$ iff $PF(p_0)=PF(q_0)$ and $p_0 \equiv_{T}^{\infty} q_0$.
\end{definition}
A process simulates another one when it is capable of mimicking every step of the simulated process such that the reached process again simulates the result of the original step.
Two processes are simulation equivalent if one simulates the other and vice versa.
\begin{definition}[Simulation Semantics]\label{def:sim}
  Let $A = ( Q, \Sigma, \transrel{})$ be an LTS.
  A {\em simulation} between $p_0$ and $q_0 \in Q$ is a binary relation $R\subseteq Q\times Q$ such that (1) $(p_0, q_0)\in R$ and (2) if $(p, q)\in R$, it holds that
  \begin{itemize}
    \item $p\transrel{a} p'$ implies that $\exists q'\in Q : q\transrel{a} q'$ and $(p',q')\in R$
  \end{itemize}
  $p_0$ and $q_0$ are {\em similar}, denoted $p_0 \seq q_0$ iff there exists a simulation relation between $p_0$ and $q_0$ and a simulation relation between $q_0$ and $p_0$.
\end{definition}
A bisimulation relation between to processes is a single relation that is a simulation for both directions simultaneously.
\begin{definition}[Bisimulation Semantics]\label{def:bisim}
  Let $A = ( Q, \Sigma, \transrel{})$ be an LTS.
  A {\em bisimulation} between $p_0$ and $q_0 \in Q$ is a binary relation $R\subseteq Q\times Q$ such that (1) $(p_0, q_0)\in R$ and (2) if $(p, q)\in R$, it holds that
  \begin{itemize}
    \item $p\transrel{a} p'$ implies that $\exists q'\in Q : q\transrel{a} q'$ and $(p',q')\in R$ and
    \item $q\transrel{a} q'$ implies that $\exists p'\in P : p\transrel{a} p'$ and $(p',q')\in R$.
  \end{itemize}
  $p_0$ and $q_0$ are {\em bisimilar}, denoted $p_0\equiv_B q_0$ iff there exists a bisimulation relation between $p_0$ and $q_0$.
\end{definition}
Corresponding processes have to be capable of the same actions leading again to equivalent processes.
All the equivalences we defined form an equivalence spectrum, having bisimilarity referring to the finest equivalence and trace equivalence being the coarsest one.
All interrelations are depicted in Fig.~\ref{fig:semeq-hier}.
For a comprehensive overview we refer to van Glabbeek~\cite{Glabbeek1990Linear}.
For common LTS operators, all the branching time equivalences described in this section are congruences, \ie equivalence of components implies equivalence of composition.
In the following sections, we analyze the notions of weak and strong fairness to define notions of {\em fair language equivalence}.
In both cases, we look at the question which of the equivalences, mentioned in this section, also implies fair language equivalence.
\section{Strong Fairness} %
\label{sec:strong-fairness}
In concurrent and distributed systems, several choices between actions take place with some of them turning up recurrently.
As those choices are resolved independently of the current execution, the system may always choose the same action while neglecting the alternatives.
Fairness assumptions are incorporated to overcome such an unfair treatment of actions.
Thereupon, according notions of fair system languages are defined.
The following definition is inspired by the notion of {\em fairness} as introduced by Reisig~\cite{Reisig1998}.
\begin{definition}[Strong Fairness]\label{def:fairness}
Let $A = (Q, \Sigma, \transrel{})$ be an LTS and $\fairatt\subseteq\Sigma$.
A run $\rho = p_0 a_1 p_1 a_2 \ldots$ of $p_0 \in Q$ is {\em strong fair \wrt $\fairatt$} iff for all $a\in\fairatt$, the existence of infinitely many $p_i$ in $\rho$ with $p_i \transrel{a}$ implies that $a$ occurs infinitely often in $\rho$.
The set of all strong fair runs of $p_0$ \wrt \fairatt is denoted by $\sfairop{\fairatt}(p_0)$.
The \emph{strong fair language of $p_0$ \wrt \fairatt} is defined as $\sFairLang{\fairAtt}(p_0) := \{\trace(\rho) \mid \rho \in \sfairop{\fairatt}(p_0) \}$.
We call a trace $\sigma$ strong fair in $p_0$ \wrt \fairatt iff $\sigma \in \sFairLang{\fairatt}(p_0)$ \ie there is a run $\rho$ with $tr(\rho) = \sigma$ that is strong fair in $p_0$ \wrt \fairatt.
\end{definition}
By this definition, finite runs are always strong fair.
Please note that the superscript $\Phi$ is just an identifier for strong fairness.
In the following section, $\varphi$ will likewise denote weak fairness. 
Fair languages describe the actually executable system behaviors, as the possibility of unfair treatment of recurrent choices vanishes in case of infinite executions.
Subsequently, we identify equivalences within the spectrum presented in Sect.~\ref{sec:preliminaries} for which equivalence of two systems implies equal strong fair languages.
We thereby find an equivalence spectrum for which strong fairness assumptions, as reflected by the choice of $\mathcal F$, leave the set of fair runs, and thereby the fair language, of a system invariant.
In a system development process, stepwise refinement \wrt this spectrum may be used to substitute components by equivalent subsystems without altering the fair behaviors of the subsystems.
\begin{figure}[tbp]
  \centering
  \begin{subfigure}[b]{.15\textwidth}
    \centering
    \begin{tikzpicture}
      \node[state,label=left:$p$] (q0) {};
      \node[state,draw=white,fill=white,below=of q0] (h) {};
      \node[state,draw=white,fill=white,below=of h] {};

      \draw[edge,loop above] (q0) to node[auto] {$a$} (q0);
      \draw[edge,loop below] (q0) to node[auto] {$b$} (q0);
    \end{tikzpicture}
  \end{subfigure}
  \qquad
  \begin{subfigure}[b]{.15\textwidth}
    \centering
    \begin{tikzpicture}
      \node[state,label=left:$q$] (q0) {};      
      \node[state,draw=white,fill=white,below=of q0] (h) {};
      \node[state,draw=white,fill=white,below=of h] {};

      \node[state] (q1) at (5em,0) {};

      \draw[edge,loop above] (q0) to node[auto] {$a$} (q0);
      \draw[edge,loop below] (q0) to node[auto] {$b$} (q0);
      \draw[edge,loop below] (q1) to node[auto] {$b$} (q1);
      \draw[edge,bend left=30] (q0) to node[auto] {$b$} (q1);
    \end{tikzpicture}
  \end{subfigure}
  \qquad
  \begin{subfigure}[b]{.25\textwidth}
    \centering
    \begin{tikzpicture}
      \node[state,label=above:$r$] (q0) {};
      \node[state, below left=of q0] (q1) {};
      \node[state, below=of q1] (q2) {};

      \draw[edge] (q0) to node[auto,swap] {$a$} (q1);
      \draw[edge] (q1) to node[auto,swap] {$b$} (q2);
      \draw[edge,bend right=10] (q2) to node[circle, below=.5em and -2em] {$c$} (q0);

      \node[state, below right=of q0] (q11) {};
      \node[state, below=of q11] (q21) {};

      \draw[edge] (q0) to node[auto] {$a$} (q11);
      \draw[edge] (q11) to node[auto] {$b$} (q21);
      \draw[edge,bend left=10] (q21) to node[circle, below=.5em and -2em] {$c$} (q0);

      \node[state, right=of q21] (q22) {};
      \node[state, below=of q22] (q32) {};

      \draw[edge] (q11) to node[auto] {$d$} (q22);
      \draw[edge] (q21) to node[auto] {$d$} (q32);

    \end{tikzpicture}
  \end{subfigure}
  \qquad
  \begin{subfigure}[b]{.25\textwidth}
    \centering
    \begin{tikzpicture}
      \node[state,label=above:$s$] (q0) {};
      \node[state, below left=of q0] (q1) {};
      \node[state, below=of q1] (q2) {};

      \draw[edge] (q0) to node[auto,swap] {$a$} (q1);
      \draw[edge] (q1) to node[auto,swap] {$b$} (q2);
      \draw[edge,bend right=10] (q2) to node[circle, below=.5em and -2em] {$c$} (q0);

      \node[state, below right=of q0] (q11) {};
      \node[state, below=of q11] (q21) {};

      \draw[edge] (q0) to node[auto] {$a$} (q11);
      \draw[edge] (q11) to node[auto] {$b$} (q21);
      \draw[edge,bend left=10] (q21) to node[circle, below=.5em and -2em] {$c$} (q0);

      \node[state, right=of q21] (q22) {};
      \node[state, below right=of q2] (q32) {};

      \draw[edge] (q11) to node[auto] {$d$} (q22);
      \draw[edge] (q2) to node[auto] {$d$} (q32);

    \end{tikzpicture}
  \end{subfigure}
  \caption{A finite LTS, where $p\teq q$, but $\sFairLang{\{a\}}(p)\neq\sFairLang{\{a\}}(q)$ and $r\req s$, but $\sFairLang{\{d\}}(r)\neq\sFairLang{\{d\}}(s)$.}
  \label{fig:traces-counterexample}
\end{figure}
As strong fairness considers the resolution of (nondeterministic) choices, a linear time equivalence such as trace equivalence should not suffice to preserve strong fairness.
This claim is supported by the example LTS depicted in Fig.~\ref{fig:traces-counterexample}.
The traces of processes $p$ and $q$ are equal, \ie $\textit{Traces}(p) = \{a, b\}^* = \textit{Traces}(q)$ and $\textit{Traces}^{\infty}(p) = \{a, b\}^\omega = \textit{Traces}^{\infty}(q)$.
However, the strong fair language \wrt $\fairAtt = \{a\}$ of both processes differ since $b^\omega$ is not fair \wrt action $a$ in $p$.
In contrast, $q$ may take the transition leading away from $q$ at some point disabling action $a$ forever and resulting in a fair trace $b^\omega$.

Ready equivalence and failures equivalence are among the weakest branching time equivalences.
They compare prefixes of possibly infinite runs and their capabilities, \ie a set of actions being enabled (readies) or disabled (failures) after the respective prefix.
But fairness also considers information on actions being enabled during the run corresponding to the mentioned prefix.
Furthermore, considering a run of a process with two prefixes $w_1$ and $w_2$ where $w_1$ is a prefix of $w_2$ does not necessarily imply the existence of a unique run with the same two prefixes in a failures/ready equivalent process.
For example, processes $r$ and $s$ in Fig.~\ref{fig:traces-counterexample} are ready equivalent thus also failures equivalent.
Both processes contain two different $abc$-loops such that action $d$ can be but does not have to be enabled after any $a$ or $b$ throughout the infinite trace $(abc)^{\omega}$.
Furthermore $s$ is forced to enable $d$ once in every $abc$ subtrace while $r$ is not.
Therefore $d$ is enabled infinitely often in any run of $s$ corresponding to $(abc)^{\omega}$ whereas $r$ can use the left loop solemnly, never enabling $d$.
As $d$ is not recurrent in $(abc)^{\omega}$, this trace is strong fair \wrt $\fairAtt=\{d\}$ in $r$ but not in $s$ resulting in different fair languages.
Thus, respecting branching time does not necessarily imply strong fair language equivalence.

However, failure trace equivalence and ready trace equivalence preserve information on disabledness or enabledness of actions on finite and infinite runs.
In contrast to failures and ready equivalence, not only capabilities at the end of prefixes are considered but also in every intermediate step.
Having this information also included in infinite runs suffices to restrict equivalent processes to those having the same strong fair language \wrt some $\fairAtt\subseteq\Sigma$.
Our proof incorporates failure trace equivalence, immediately implying the result for stronger equivalences, \eg ready trace equivalence.
\begin{theorem}\label{thm:failuretraceequiv-fairness}
  Let $A = (Q, \Sigma, \transrel{})$ be an LTS and $p_0, q_0 \in Q$ such that $p_0\failtraceeq q_0$.
  For any $\fairAtt\subseteq\Sigma$, $\sFairLang{\fairAtt}(p_0)=\sFairLang{\fairAtt}(q_0)$.
\end{theorem}
\begin{proof}
  Assume, there is some trace $\sigma\in\sFairLang{\fairAtt}(p_0)\setminus\sFairLang{\fairAtt}(q_0)$.
  It holds that $\sigma\in\textit{Traces}^{\infty}(q_0)$, as
  \begin{inparaenum}[(a)]
  \item if $\sigma$ was finite, it would be in $\sFairLang{\fairAtt}(q_0)$ and
  \item from $p_0\failtraceeq q_0$ we deduce $p_0\equiv_T^{\infty} q_0$, implying $\textit{Traces}^{\infty}(p_0)=\textit{Traces}^{\infty}(q_0)$.
  \end{inparaenum}
  As $\sigma\in\sFairLang{\fairAtt}(p_0)$, there is a strong fair run $\rho_p = p_0 a_1 p_1 a_2 \ldots$ \wrt $\fairAtt$ such that $\trace(\rho_p) = \sigma$.
  Let $\sigma_p$ be the maximal failure trace corresponding to $\rho_p$ \ie the trace that alternates between the action labels of $\rho_p$ and the maximal refusal sets in each state.
  This $\sigma_p$ is defined as $\sigma_p = X_0 a_1 X_1 a_2 \dots$ where $X_i = \{ a' \in \Sigma \mid p_i\not\transrel{a'} \}$. As $\sigma_p$ corresponds to a strong fair run of $p_0$, it holds that for each $a\in\fairAtt$,
\begin{inparaenum}[(i)]
\item label $a$ occurs infinitely often in $\rho_p$, \ie $\forall i\in\mathbb{N}: \exists j>i: a_j = a$ or
\item the number of refusal sets that the label $a$ does not occur in is finite, \ie $\exists n\in\mathbb{N}:\forall j \geq n: a \in X_j$.
\end{inparaenum}
From $p_0\failtraceeq q_0$ we know that $\sigma_p$ is also a failure trace of $q_0$.
Though $\sigma_p$ is a maximal failure trace for $p_0$ it does not have to be maximal for $q_0$.
A corresponding maximal failure trace for $q_0$ would be $\sigma_q = Y_0 a_1 Y_1 a_2 \dots$ with refusal sets $Y_i = \{ a' \in \Sigma \mid q_i\not\transrel{a'} \}$ and we get that $X_i \subseteq Y_i$.
As $\sigma \not\in \sFairLang{\fairAtt}(q_0)$, every run $\rho_q = q_0 a_1 q_1 a_2 \ldots$ that corresponds to $\sigma_q$ has to neglect strong fairness for some $a\in\fairAtt$.
If case (i) holds, $\rho_q$ contains $a$ infinitely often and, by this, respects strong fairness of $a$.
If case (ii) holds, there is an $n\in\mathbb{N}$ such that every refusal set $Y_j = \{ a' \in \Sigma \mid q_j\not\transrel{a'} \}$ with $j>n$ contains $a$.
Otherwise $\rho_q$ would not be a run that corresponds to $\sigma_q$. Therefore $\rho_q$ does not neglect fairness of $a$ which contradicts $\sigma\not\in \sFairLang{\fairatt}(q_0)$.
The case of $\sigma\in\sFairLang{\fairAtt}(q_0)\setminus\sFairLang{\fairAtt}(p_0)$ is analogous.
\end{proof}
Our result immediately implies that all equivalences $\equiv\subseteq\failtraceeq$ preserve strong fairness.
This result together with the following counterexamples, concerning ready equivalence, possible futures equivalence and simulation equivalence, indicate a spectrum of strong fairness preserving equivalences (\cf Fig.~\ref{fig:semeq-hier}).

As pointed out in the discussion about trace equivalence, processes $p$ and $q$ in Fig.~\ref{fig:traces-counterexample} do not share the same strong fair language.
However, it holds that $p \seq q$ which may be seen by simulation relations containing nothing but the pair $(p,q)$ ($(q,p)$, resp.).
Every $p \transrel{a} p$ step can be simulated by a corresponding $q \transrel{a} q$ step and vice versa.
The same holds for $p \transrel{b} p$ steps.
The $b$ step of $q$ leading away from $q$ and all subsequent $b$ steps can also be simulated by a $p \transrel{b} p$ step.
This counterexample lets us conclude that simulation semantics does not preserve fairness.

\begin{figure}[htbp]
  \centering
  \begin{subfigure}{.4\linewidth}
    \centering
    \begin{tikzpicture}
		  \node[state,label=above:$p$] (q0) {};
		  \node[state, left=of q0] (q1) {};
		  \node[state, right=of q0] (q5) {};
		  \node[state, right=of q5] (q6) {};
		
		  \draw[edge,bend right=30] (q0) to node[auto,swap] {$a$} (q1);
          \draw[edge,bend right=30] (q1) to node[auto,swap] {$a$} (q0);
          \draw[edge,bend left=30] (q0) to node[auto] {$a$} (q5);
          \draw[edge,bend left=30] (q5) to node[auto] {$a$} (q0);
          \draw[edge] (q5) to node[auto] {$c$} (q6);
    \end{tikzpicture}
  \end{subfigure}
  \qquad
  \begin{subfigure}{.4\linewidth}
    \centering
		\begin{tikzpicture}
		  \node[state,label=above:$q$] (q0) {};
		  \node[state, left=of q0] (q1) {};
		  \node[state, below=of q1] (q2) {};
		  \node[state, right=of q2] (q3) {};
		  \node[state, below=of q3] (q4) {};
		  \node[state, right=of q0] (q5) {};
		  \node[state, right=of q5] (q6) {};
		
		  \draw[edge] (q0) to node[auto,swap] {$a$} (q1);
          \draw[edge] (q1) to node[auto,swap] {$a$} (q2);
          \draw[edge] (q2) to node[auto,swap] {$a$} (q3);
          \draw[edge] (q3) to node[auto] {$a$} (q0);
          \draw[edge] (q3) to node[auto,swap] {$c$} (q4);
          \draw[edge,bend left=30] (q0) to node[auto] {$a$} (q5);
          \draw[edge,bend left=30] (q5) to node[auto] {$a$} (q0);
          \draw[edge] (q5) to node[auto] {$c$} (q6);
		\end{tikzpicture}
  \end{subfigure}
  \caption{
  A finite LTS, where $p\equiv_{PF}^\infty q$, but $\sFairLang{\{c\}}(p)\neq\sFairLang{\{c\}}(q)$.
  }
  \label{fig:possible-futures-counterexample}
\end{figure}
The LTS shown in Fig.~\ref{fig:possible-futures-counterexample} has two possible futures equivalent processes $p$ and $q$ that do have different fair languages \wrt $\fairAtt = \{c\}$.
Process $p$ can take the left $a$ loop to produce an infinite trace $a^\omega$ whereas $q$ is not capable of doing infinite $a$ actions without enabling $c$ infinitely often and by this treating $c$ unfair.
\section{Weak Fairness} %
\label{sec:weak-fairness}
The second characteristic of fairness we discuss is weak fairness.
In contrast to strong fairness, it does not suffice for an action to be enabled infinitely often to enforce an infinite number of occurrences in a fair run.
In weak fair runs, every action that is enabled infinitely long has to be taken infinitely often.
The following definition is inspired by the notion of {\em progress} as introduced by Reisig~\cite{Reisig1998}.
\begin{definition}[Weak Fairness]\label{def:weak-fairness}
Let $A = (Q, \Sigma, \transrel{})$ be an LTS and $\fairatt\subseteq\Sigma$.
A run $\rho = p_0 a_1 p_1 a_2 \ldots$ of $p_0 \in Q$ is {\em weak fair \wrt $\fairatt$} iff for all $a\in\fairatt$, the existence of an $i \geq 0$ with $p_j \transrel{a}$ for all $j \geq i$ implies that $a$ occurs infinitely often in $\rho$.
The set of all weak fair runs of $p_0$ \wrt \fairatt is denoted by $\wfairop{\fairatt}(p_0)$.
The \emph{weak fair language of $p_0$ \wrt \fairatt} is defined as $\wFairLang{\fairAtt}(p_0) := \{\trace(\rho) \mid \rho \in \wfairop{\fairatt}(p_0) \}$.
We call a trace $\sigma$ weak fair in $p_0$ \wrt \fairatt iff $\sigma \in \wFairLang{\fairatt}(p_0)$ \ie there is a run $\rho$ with $tr(\rho) = \sigma$ that is weak fair in $p_0$ \wrt \fairatt.
\end{definition}
Note that finite runs are always weak fair.
In case of preservation of weak fairness under ready equivalence we again find a counterexample.
The processes $p$ and $q$ depicted in Fig.~\ref{fig:weak-fairness-counterexample} are ready equivalent but their weak fair languages \wrt $\fairAtt = \{d, e\}$ differ.
The trace $a(bc)^\omega$ has no according fair run in $p$ as each run producing this trace would either have $d$ or $e$ enabled in each step following the first $a$ step.
In contrast, $q$ alternates between enabling $d$ and $e$ throughout each run.
Thus, weak fairness of $\{d, e\}$ is not neglected by $q$ and $\wFairLang{\{d, e\}}(p) \not\ni a(bc)^\omega \in \wFairLang{\{d, e\}}(q)$.
\begin{figure}[tbp]
  \centering
  \begin{subfigure}{.4\linewidth}
    \centering
    \begin{tikzpicture}
		  \node[state,label=above:$p$] (q0) {};
		  \node[state, below left=of q0] (q1) {};
		  \node[state, below=of q1] (q2) {};
		
		  \draw[edge] (q0) to node[auto,swap] {$a$} (q1);
		  \draw[edge,bend left=30] (q1) to node[auto] {$b$} (q2);
		  \draw[edge,bend left=30] (q2) to node[auto] {$c$} (q1);
		
		  \node[state, below right=of q0] (q11) {};
		  \node[state, below=of q11] (q21) {};

		  \draw[edge] (q0) to node[auto] {$a$} (q11);
		  \draw[edge,bend right=30] (q11) to node[auto,swap] {$b$} (q21);
		  \draw[edge,bend right=30] (q21) to node[auto,swap] {$c$} (q11);

		  \node[state, above right=of q11] (q22) {};
		  \node[state, below right=of q21] (q32) {};
		
		  \draw[edge] (q11) to node[auto,swap] {$e$} (q22);
		  \draw[edge] (q21) to node[auto] {$e$} (q32);
		
		  \node[state, above left=of q1] (q22) {};
		  \node[state, below left=of q2] (q32) {};
		
		  \draw[edge] (q1) to node[auto] {$d$} (q22);
		  \draw[edge] (q2) to node[auto,swap] {$d$} (q32);
    \end{tikzpicture}
  \end{subfigure}
  \qquad
  \begin{subfigure}{.4\linewidth}
    \centering
		\begin{tikzpicture}
		  \node[state,label=above:$q$] (q0) {};
		  \node[state, below left=of q0] (q1) {};
		  \node[state, below=of q1] (q2) {};
		
		  \draw[edge] (q0) to node[auto,swap] {$a$} (q1);
		  \draw[edge,bend left=30] (q1) to node[auto] {$b$} (q2);
		  \draw[edge,bend left=30] (q2) to node[auto] {$c$} (q1);
		
		  \node[state, below right=of q0] (q11) {};
		  \node[state, below=of q11] (q21) {};

		  \draw[edge] (q0) to node[auto] {$a$} (q11);
		  \draw[edge,bend right=30] (q11) to node[auto,swap] {$b$} (q21);
		  \draw[edge,bend right=30] (q21) to node[auto,swap] {$c$} (q11);

		  \node[state, above right=of q11] (q22) {};
		  \node[state, below right=of q21] (q32) {};
		
		  \draw[edge] (q11) to node[auto,swap] {$e$} (q22);
		  \draw[edge] (q21) to node[auto] {$d$} (q32);
		
		  \node[state, above left=of q1] (q22) {};
		  \node[state, below left=of q2] (q32) {};
		
		  \draw[edge] (q1) to node[auto] {$d$} (q22);
		  \draw[edge] (q2) to node[auto,swap] {$e$} (q32);
		
		\end{tikzpicture}
  \end{subfigure}
  \caption{
  A finite LTS, where $p\req q$, but $\wFairLang{\{d,e\}}(p)\neq\wFairLang{\{d,e\}}(q)$.
  }
  \label{fig:weak-fairness-counterexample}
\end{figure}

We may show that failure trace equivalence preserves also weak fairness, leaving us with the same spectrum of weak fairness preserving equivalences as in case of strong fairness.
\begin{theorem}\label{thm:failuretraceequiv-weakfairness}
  Let $A = (Q, \Sigma, \transrel{})$ be an LTS and $p_0, q_0 \in Q$ such that $p_0\failtraceeq q_0$.
  For any $\fairAtt\subseteq\Sigma$, $\wFairLang{\fairAtt}(p_0)=\wFairLang{\fairAtt}(q_0)$.
\end{theorem}

\begin{proof}
  Assume there is a trace $\sigma\in\wFairLang{\fairAtt}(p_0)\setminus\wFairLang{\fairAtt}(q_0)$.
  It holds that $\sigma\in\textit{Traces}^{\infty}(q_0)$, since
  \begin{inparaenum}[(a)]
  \item if $\sigma$ was finite, it would be in $\wFairLang{\fairAtt}(p_0)$ and
  \item $p_0\failtraceeq q_0$ implies that $p_0\equiv_T^{\infty} q_0$, implying $\textit{Traces}^{\infty}(p_0)=\textit{Traces}^{\infty}(q_0)$.
  \end{inparaenum}
  As $\sigma\in\wFairLang{\fairAtt}(p_0)$, there is a weak fair run $\rho_p = p_0 a_1 p_1 a_2 \ldots$ \wrt $\fairAtt$ such that $\textit{tr}(\rho_p) = \sigma$.
  Let $\sigma_p$ be the maximal failure trace corresponding to $\rho_p$ \ie the trace that alternates between the action labels of $\rho_p$ and the maximal refusal sets in each state.
  This $\sigma_p$ is defined as $\sigma_p = X_0 a_1 X_1 a_2 \dots$ where $X_i = \{ a' \in \Sigma \mid p_i\not\transrel{a'} \}$. 
  As $\sigma_P$ corresponds to a weak fair run of $p_0$, it holds that for each $a\in\fairAtt$,
\begin{inparaenum}[(i)]
\item the label $a$ occurs infinitely often in $\rho_p$, \ie $\forall i\in\mathbb{N}: \exists j>i: a_j = a$, or
\item 
the label $a$ occurs in infinitely many refusal sets, \ie for every $n$ there is $j \geq n$ such that $X_j$ contains $a$. This can be formalized by $\forall n\in\mathbb{N}:\exists j \geq n: a \in X_j$.
\end{inparaenum}

From $p_0\failtraceeq q_0$ we know that $\sigma_p$ is also a failure trace of $q_0$.
Though $\sigma_p$ is a maximal failure trace of $p_0$ it does not have to a be maximal of $q_0$.
A corresponding maximal failure trace of $q_0$ would be $\sigma_q = Y_0 a_1 Y_1 a_2 \dots$ with refusal sets $Y_i = \{ a' \in \Sigma \mid q_i\not\transrel{a'} \}$ and we get that $X_i \subseteq Y_i$.
As $\sigma \not\in \wFairLang{\fairAtt}(q_0)$, every run $\rho_q = q_0 a_1 q_1 a_2 \ldots$ that corresponds to $\sigma_q$ neglects fairness for some $a\in\fairAtt$.
If case (i) holds, $\rho_q$ contains $a$ infinitely often and, by this, respects weak fairness of $a$.
In case (ii), for every $n\in\mathbb{N}$ there is a $j \geq n$ such that $Y_j = \{ a' \in \Sigma \mid q_j\not\transrel{a'} \}$ contains $a$.
This can directly be followed from $X_i \subseteq Y_i$.
Otherwise $\rho_q$ would not be a run that corresponds to $\sigma_q$. Therefore $\rho_q$ does not neglect weak fairness of $a$, contradicting $\sigma\not\in \wFairLang{\fairAtt}(q_0)$.
The case of $\sigma\in\wFairLang{\fairAtt}(q_0)\setminus\wFairLang{\fairAtt}(p_0)$ is analogous.
\end{proof}
The crucial point in proving preservation of weak and strong fairness under failure trace equivalence is the construction of a maximal failure trace that corresponds to a given generic trace.
This construct contains sufficient information to decide whether the given generic trace is fair regardless of weak or strong fairness.
As those traces are preserved for failure trace equivalent systems, fairness is preserved.
The proofs of strong and weak fairness only differ in the description of disallowed enabledness for actions in $\fairAtt\subseteq\Sigma$ that are not recurrent.
While strong fairness requires the actions to be always eventually disabled, weak fairness needs those actions to be eventually always disabled.

As in Sect.~\ref{sec:strong-fairness}, we also discuss the weak fair languages of simulation equivalent as well as possible futures equivalent processes. 
Regarding simulation equivalence, processes $p$ and $q$ of Fig.~\ref{fig:traces-counterexample} again serve as a counterexample, since they are simulation equivalent but have different weak fair languages.
This counterexample lets us conclude that simulation semantics preserves neither strong nor weak fairness.
Also in case of possible futures we subsequently discuss a counterexample justifying the indicated border of fairness preservation in Fig.~\ref{fig:semeq-hier}. 
The LTS shown in Fig.~\ref{fig:weak-fairness-pf-counterexample} has two possible futures equivalent processes $p$ and $q$ that do show differences in their weak fair languages \wrt $\fairAtt = \{b\}$.
Process $p$ may follow the infinite branch to produce the infinite trace $a^\omega$ thereby respecting weak fairness of action $b$, whereas every trace $a^\omega$ of $q$ eventually enables $b$ in every subsequent step.
Each trace of $p$ returning to $p$ (always after $2n$ steps for some $n\in\mathbb N$) is reflected by infinite branching of $q$ into an infinite number of paths of finite lengths of multiples of two, indicated by the three horizontally (process $p$) or vertically (process $q$) aligned dots ($\cdots$).

\begin{figure}[tbp]
  \centering
  \begin{subfigure}{.4\linewidth}
    \centering
    \begin{tikzpicture}
		  \node[state,label=left:$p$] (q0) {};
		  \node[state, right=of q0] (q1) {};
		  \node[state, right=of q1] (q2) {};
		  \node[state, right=of q2] (q3) {};
		
          \draw[edge,loop above] (q0) to node[auto] {$a$} (q0);
          \draw[edge,loop below] (q0) to node[auto] {$b$} (q0);
		  \draw[edge, bend left = 30] (q0) to node[auto] {$a$} (q1);
          \draw[edge, bend left = 30] (q1) to node[auto] {$a$} (q2);
          \draw[edge, bend left = 30] (q2) to node[auto] {$a$} (q3);
          \draw[edge, bend left = 30] (q1) to node[auto] {$a$} (q0);
          \draw[edge, bend left = 30] (q2) to node[auto] {$a$} (q1);
          \draw[edge, bend left = 30] (q3) to node[auto] {$a$} (q2);
          \node[right=1em of q3] () {$\hdots$};

    \end{tikzpicture}
  \end{subfigure}
  \qquad
  \begin{subfigure}{.4\linewidth}
    \centering
		\begin{tikzpicture}
		  \node[state,label=left:$q$] (q0) {};
          \node[below=of q0, draw=none] (al1) {};
          \node[state,right=of al1] (qaa1) {}
          	edge[pre] node[auto,swap] {$a$} (q0);
          \node[state,right=of qaa1] (qaa2) {}
          	edge[pre] node[auto,swap] {$a$} (qaa1);
          \node[state,right=of qaa2] (qaa3) {}
          	edge[pre] node[auto,swap] {$a$} (qaa2);
          \node[state,above=of qaa2] (qa1) {}
          	edge[pre] node[auto,swap] {$a$} (q0);
          \node[right=of qaa3, draw=none] (al2) {};
		  \node[state,fill=white,draw, above=of al2] (q2) {}
          	edge[pre,bend right=25] node[auto,swap] {$a$} (q0)
          	edge[pre,bend right=50] node[auto,swap] {$b$} (q0)
          	edge[pre] node[auto,swap] {$a$} (qa1)
          	edge[pre] node[auto,swap] {$a$} (qaa3);
          \node[below=.7em of qaa2] {$\vdots$};
		  \draw[edge,loop above] (q2) to node[auto] {$a$} (q2);
          \draw[edge,loop below] (q2) to node[auto] {$b$} (q2);

		\end{tikzpicture}
  \end{subfigure}
  \caption{
  An LTS, where $p\equiv_{PF}^\infty q$, but $\wFairLang{\{b\}}(p)\neq\wFairLang{\{b\}}(q)$.
  }
  \label{fig:weak-fairness-pf-counterexample}
\end{figure}

The reason why the processes depicted in Fig.~\ref{fig:weak-fairness-pf-counterexample} are indistinguishable by possible futures is that the captured futures only account for the traces from a certain state and traces are, by definition, finite.
When considering a slightly stronger notion of possible futures, incorporating also infinite traces as possible futures, $p$ and $q$ may be distinguished, possibly altering the decision on weak fairness preservation in general.
Further observe that the given processes both have an infinite state space, which is different from any other counterexample given in this paper. 
We do not give a formal proof but strongly conjecture that for finite-state processes (\ie for which the set of all reachable processes is finite), possible futures equivalence preserves weak fairness.
The idea of the proof is that a weak fair run of a finite-state process $p$ traverses at least one reachable process infinitely many times.
Moreover, there must be such a process respecting weak fairness of some action $b$, \ie $b$ is disabled and the trace $b$ is no possible future of that process.
Such a process must also be reachable from a possible futures equivalent process $q$, which in turn may be used to construct the same weak fair run as given by $p$.
Since our argument is concerned with an arbitrary infinite weak fair run, the claim holds in general.
The counterexample for strong fairness, given by Fig.~\ref{fig:possible-futures-counterexample}, already employs finite-state processes thus indicating a split of the borders of fairness preservation for finite-state processes.

\section{Conclusion} %
\label{sec:related-work}
In this paper, we discovered an equivalence spectrum for which fairness assumptions are preserved between equivalent systems.
When distinguishing between internal and external moves of a system, as usually done in process-algebraic verification, handling internal actions in a fair way has already been studied for notions of global fairness~\cite{Bergstra1987,Puhakka2001Liveness,Puhakka2005Using}, \ie where not only a subset of actions and/or components is considered.
The main focus of these works is the proper handling of divergent system behavior, manifested in infinite sequences of internal actions, \ie $\tau$.
Thereupon, {\em Koomen's Fair Abstraction Rule} (KFAR)~\cite{Koomen85} is taken into account allowing to reduce divergent behavior to a single internal action in process-algebraic settings.
The resulting semantic equivalence Bergstra \etal~\cite{Bergstra1987} discovered and Puhakka and Valmari extensively studied~\cite{Puhakka2001Liveness,Puhakka2005Using} is called {\em Chaos-Free Failures Divergences} (CFFD).

There are open questions we plan to address in future work.
First, in this work we dealt with equivalence relations not abstracting from internal behavior (\ie $\tau$-transitions) and, thereupon, also fairness of actions that are not $\tau$.
When trying to lift our results to a weak equivalence spectrum~\cite{Glabbeek1993}, it is unclear what treatment of $\tau$-transitions is preferable.
Previous works, as mentioned above, already present a variety of possibilities, each worth investigating.
Besides the inclusion of internal actions, as conjectured in the last section, considering only finite state systems may yield even more diversity between the equivalences preserving different styles of fairness.

Furthermore, we would like to explore the conjecture, given in the end of Sect.~\ref{sec:weak-fairness}, namely that possible futures equivalence implies weak fair language equivalence in case of finite-state processes.
It might be interesting to see how the spectrum of fairness preservation evolves for the various notions of progress in finite-state processes.

Another important aspect for future work is the composition of system behaviors.
Some of the equivalences in van Glabbeek's spectrum are congruences for certain composition operators.
Congruences allow for a compositional reasoning in the sense that congruent systems show the same overall behavior when plugged into a fixed environment.
If for two systems, correctness is proven under the same fairness assumption, it is unclear whether the composed system behavior under fairness only constitutes the fair behaviors of the components.

\begin{figure}[htbp]
  \centering
  \begin{tikzpicture}
    \node[state,label=left:$p$] (q0) at (-2,-3) {};
    \draw[edge,loop below] (q0) to node[auto] {$\textit{a}$} (q0);
    \draw[edge,loop above] (q0) to node[auto] {$\textit{b}$} (q0);

    \node[state,label=left:$q$] (p0) at (0, -3) {};
    \node[state,right=of p0] (p1) {};
    \draw[edge,bend left=30] (p1) to node[auto] {$\textit{a}$} (p0);
    \draw[edge,bend left=30] (p0) to node[auto] {$\textit{c}$} (p1);
    \node at (-1, -3) {\Large $\parll$};

    \node[right=of p1] (eq) {\Large$\leadsto$};
    \node[right=of eq] (eq) {};
    \node[state,label=left:$(p\comma q)$,right=of eq] (r0) {};
    \node[right=of r0] (r01) {};
    \node[node distance=1em,state,right=of r01] (r1) {};
    \draw[edge,loop above] (r0) to node[auto] {$\textit{b}$} (r0);
    \draw[edge,loop above] (r1) to node[auto] {$\textit{b}$} (r1);
    \draw[edge,bend left=30] (r0) to node[auto] {$\textit{c}$} (r1);
    \draw[edge,bend left=30] (r1) to node[auto] {$\textit{a}$} (r0);
  \end{tikzpicture}
  \caption{
    Two systems $p$ and $q$ and a possible composition $(p,q)$.
  }
  \label{fig:composition-example}
\end{figure}
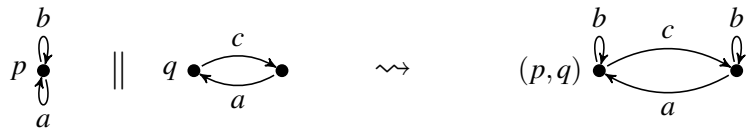

Consider for example the systems $p$ and $q$ in Figure \ref{fig:composition-example} under fairness of $\fairAtt = \{a\}$.
Regardless whether strong or weak fairness is taken into account, the run $b^\omega$ is no fair run of process $p$.
But when composed with process $q$ in a manner where transitions with the same action labels have to be done synchronously, the disregard of joined action $a$ becomes a possibility as it is not enabled initially.
Therefore, $b^\omega$ becomes a fair trace of the composed system whereas it does not occur in any of the fair languages of the components.
This example highlights the influence of the used composition mechanism and properties of components on the relation of fair languages of components and composed system.

When considering composition operators like the one of CCS by Milner~\cite{Milner1982}, dealing with internal behavior under fairness is crucial.
Puhakka and Valmari~\cite{Puhakka2001Liveness} studied this subject for CFFD in a general LTS composition framework.
Thereby, they limited the notions of fairness in order to obtain a usable abstraction and verification framework.
Stronger equivalences as the ones we identified in this paper have not been considered so far.
Other works are also concerned directly or indirectly with different fairness notions in process algebra settings~\cite{Costa1984Fair,Corradini2006Fairness,Francez1986}, \eg Corradini \etal~\cite{Corradini2006Fairness} consider fairness of actions.
They obtain a compositional semantics for the process language PAFAS, incorporating a TCSP-parallel operator and ensuring weak fairness by forcing each enabled action to happen eventually.

\bibliographystyle{eptcs}

\begin{thebibliography}{10}
\providecommand{\bibitemdeclare}[2]{}
\providecommand{\surnamestart}{}
\providecommand{\surnameend}{}
\providecommand{\urlprefix}{Available at }
\providecommand{\url}[1]{\texttt{#1}}
\providecommand{\href}[2]{\texttt{#2}}
\providecommand{\urlalt}[2]{\href{#1}{#2}}
\providecommand{\doi}[1]{doi:\urlalt{http://dx.doi.org/#1}{#1}}
\providecommand{\bibinfo}[2]{#2}

\bibitemdeclare{inproceedings}{Bergstra1987}
\bibitem{Bergstra1987}
\bibinfo{author}{Jan~A. \surnamestart Bergstra\surnameend},
  \bibinfo{author}{Jan~W. \surnamestart Klop\surnameend} \&
  \bibinfo{author}{Ernst-Rüdiger \surnamestart Olderog\surnameend}
  (\bibinfo{year}{1987}): \emph{\bibinfo{title}{Failures without Chaos: a
  Process Semantics for Fair Abstraction}}.
\newblock In \bibinfo{editor}{M.~\surnamestart Wirsing\surnameend}, editor:
  {\sl \bibinfo{booktitle}{Formal Description of Programming Concepts -- III}},
  \bibinfo{series}{Lecture Notes in Computer Science},
  \bibinfo{publisher}{North-Holland}, \bibinfo{address}{Amsterdam}, pp.
  \bibinfo{pages}{77--101}.

\bibitemdeclare{article}{Corradini2006Fairness}
\bibitem{Corradini2006Fairness}
\bibinfo{author}{Flavio \surnamestart Corradini\surnameend},
  \bibinfo{author}{Maria~R. \surnamestart Di~Berardini\surnameend} \&
  \bibinfo{author}{Walter \surnamestart Vogler\surnameend}
  (\bibinfo{year}{2006}): \emph{\bibinfo{title}{Fairness of Actions in System
  Computations}}.
\newblock {\sl \bibinfo{journal}{Acta Informatica}}
  \bibinfo{volume}{43}(\bibinfo{number}{2}), pp. \bibinfo{pages}{73--130},
  \doi{10.1007/s00236-006-0011-2}.

\bibitemdeclare{article}{Costa1984Fair}
\bibitem{Costa1984Fair}
\bibinfo{author}{Gerardo \surnamestart Costa\surnameend} \&
  \bibinfo{author}{Colin \surnamestart Stirling\surnameend}
  (\bibinfo{year}{1984}): \emph{\bibinfo{title}{A fair calculus of
  communicating systems}}.
\newblock {\sl \bibinfo{journal}{Acta Informatica}}
  \bibinfo{volume}{21}(\bibinfo{number}{5}), pp. \bibinfo{pages}{417--441},
  \doi{10.1007/BF00271640}.

\bibitemdeclare{article}{Dijkstra65}
\bibitem{Dijkstra65}
\bibinfo{author}{Edsger~W. \surnamestart Dijkstra\surnameend}
  (\bibinfo{year}{1965}): \emph{\bibinfo{title}{Solution of a problem in
  concurrent programming control}}.
\newblock {\sl \bibinfo{journal}{{CACM}}}
  \bibinfo{volume}{8}(\bibinfo{number}{9}), p. \bibinfo{pages}{569},
  \doi{10.1145/365559.365617}.

\bibitemdeclare{book}{Francez1986}
\bibitem{Francez1986}
\bibinfo{author}{Nissim \surnamestart Francez\surnameend}
  (\bibinfo{year}{1986}): \emph{\bibinfo{title}{Fairness}}.
\newblock \bibinfo{publisher}{Springer-Verlag New York, Inc.},
  \bibinfo{address}{New York, NY, USA}, \doi{10.1007/978-1-4612-4886-6}.

\bibitemdeclare{inproceedings}{Glabbeek1990Linear}
\bibitem{Glabbeek1990Linear}
\bibinfo{author}{Rob~J. \surnamestart van Glabbeek\surnameend}
  (\bibinfo{year}{1990}): \emph{\bibinfo{title}{The linear time - branching
  time spectrum}}.
\newblock In \bibinfo{editor}{J.~C.~M. \surnamestart Baeten\surnameend} \&
  \bibinfo{editor}{J.~W. \surnamestart Klop\surnameend}, editors: {\sl
  \bibinfo{booktitle}{CONCUR '90 Theories of Concurrency: Unification and
  Extension: Amsterdam, The Netherlands, August 27--30, 1990 Proceedings}},
  \bibinfo{publisher}{Springer Berlin Heidelberg}, \bibinfo{address}{Berlin,
  Heidelberg}, pp. \bibinfo{pages}{278--297}, \doi{10.1007/BFb0039066}.

\bibitemdeclare{inproceedings}{Glabbeek1993}
\bibitem{Glabbeek1993}
\bibinfo{author}{Rob~J. \surnamestart van Glabbeek\surnameend}
  (\bibinfo{year}{1993}): \emph{\bibinfo{title}{The linear time --- Branching
  time spectrum II}}.
\newblock In \bibinfo{editor}{Eike \surnamestart Best\surnameend}, editor: {\sl
  \bibinfo{booktitle}{CONCUR'93: 4th Intrenational Conference on Concurrency
  Theory Hildesheim, Germany, August 23--26, 1993 Proceedings}},
  \bibinfo{publisher}{Springer Berlin Heidelberg}, \bibinfo{address}{Berlin,
  Heidelberg}, pp. \bibinfo{pages}{66--81}, \doi{10.1007/3-540-57208-2_6}.

\bibitemdeclare{article}{Hennessy85}
\bibitem{Hennessy85}
\bibinfo{author}{Matthew \surnamestart Hennessy\surnameend} \&
  \bibinfo{author}{Robin \surnamestart Milner\surnameend}
  (\bibinfo{year}{1985}): \emph{\bibinfo{title}{Algebraic Laws for
  Nondeterminism and Concurrency}}.
\newblock {\sl \bibinfo{journal}{J. ACM}}
  \bibinfo{volume}{32}(\bibinfo{number}{1}), pp. \bibinfo{pages}{137--161},
  \doi{10.1145/2455.2460}.

\bibitemdeclare{article}{Knuth66}
\bibitem{Knuth66}
\bibinfo{author}{Donald~E. \surnamestart Knuth\surnameend}
  (\bibinfo{year}{1966}): \emph{\bibinfo{title}{Additional comments on a
  problem in concurrent programming control}}.
\newblock {\sl \bibinfo{journal}{Commun. {ACM}}}
  \bibinfo{volume}{9}(\bibinfo{number}{5}), pp. \bibinfo{pages}{321--322},
  \doi{10.1145/355592.365595}.

\bibitemdeclare{article}{Koomen85}
\bibitem{Koomen85}
\bibinfo{author}{Cees~J. \surnamestart Koomen\surnameend}
  (\bibinfo{year}{1985}): \emph{\bibinfo{title}{Algebraic specification and
  verification of communication protocols}}.
\newblock {\sl \bibinfo{journal}{Science of Computer Programming}}
  \bibinfo{volume}{5}, pp. \bibinfo{pages}{1 -- 36},
  \doi{10.1016/0167-6423(85)90002-4}.

\bibitemdeclare{book}{Milner1982}
\bibitem{Milner1982}
\bibinfo{author}{Robin \surnamestart Milner\surnameend} (\bibinfo{year}{1980}):
  \emph{\bibinfo{title}{A Calculus of Communicating Systems}}.
\newblock \bibinfo{publisher}{Springer Berlin Heidelberg},
  \bibinfo{address}{Berlin, Heidelberg}, \doi{10.1007/3-540-10235-3}.

\bibitemdeclare{inproceedings}{Moller99}
\bibitem{Moller99}
\bibinfo{author}{Faron \surnamestart Moller\surnameend} \&
  \bibinfo{author}{Alexander \surnamestart Rabinovich\surnameend}
  (\bibinfo{year}{1999}): \emph{\bibinfo{title}{On the expressive power of
  CTL}}.
\newblock In: {\sl \bibinfo{booktitle}{Proceedings. 14th Symposium on Logic in
  Computer Science (Cat. No. PR00158)}}, pp. \bibinfo{pages}{360--368},
  \doi{10.1109/LICS.1999.782631}.

\bibitemdeclare{inbook}{Plotkin1982}
\bibitem{Plotkin1982}
\bibinfo{author}{Gordon~David \surnamestart Plotkin\surnameend}
  (\bibinfo{year}{1982}): \emph{\bibinfo{title}{A powerdomain for countable
  non-determinism}}, pp. \bibinfo{pages}{418--428}.
\newblock \bibinfo{publisher}{Springer Berlin Heidelberg},
  \bibinfo{address}{Berlin, Heidelberg}, \doi{10.1007/BFb0012788}.

\bibitemdeclare{inproceedings}{Puhakka2005Using}
\bibitem{Puhakka2005Using}
\bibinfo{author}{Antti \surnamestart Puhakka\surnameend}
  (\bibinfo{year}{2005}): \emph{\bibinfo{title}{Using Fairness Constraints in
  Process-Algebraic Verification}}.
\newblock In \bibinfo{editor}{Dang \surnamestart Van~Hung\surnameend} \&
  \bibinfo{editor}{Martin \surnamestart Wirsing\surnameend}, editors: {\sl
  \bibinfo{booktitle}{Theoretical Aspects of Computing -- ICTAC 2005: Second
  International Colloquium, Hanoi, Vietnam, October 17-21, 2005. Proceedings}},
  \bibinfo{publisher}{Springer Berlin Heidelberg}, \bibinfo{address}{Berlin,
  Heidelberg}, pp. \bibinfo{pages}{546--561}, \doi{10.1007/11560647_36}.

\bibitemdeclare{inproceedings}{Puhakka2001Liveness}
\bibitem{Puhakka2001Liveness}
\bibinfo{author}{Antti \surnamestart Puhakka\surnameend} \&
  \bibinfo{author}{Antti \surnamestart Valmari\surnameend}
  (\bibinfo{year}{2001}): \emph{\bibinfo{title}{Liveness and Fairness in
  Process-Algebraic Verification}}.
\newblock In \bibinfo{editor}{Kim~G. \surnamestart Larsen\surnameend} \&
  \bibinfo{editor}{Mogens \surnamestart Nielsen\surnameend}, editors: {\sl
  \bibinfo{booktitle}{CONCUR 2001 --- Concurrency Theory: 12th International
  Conference Aalborg, Denmark, August 20--25, 2001 Proceedings}},
  \bibinfo{publisher}{Springer Berlin Heidelberg}, \bibinfo{address}{Berlin,
  Heidelberg}, pp. \bibinfo{pages}{202--217}, \doi{10.1007/3-540-44685-0_14}.

\bibitemdeclare{book}{Reisig1998}
\bibitem{Reisig1998}
\bibinfo{author}{Wolfgang \surnamestart Reisig\surnameend}
  (\bibinfo{year}{1998}): \emph{\bibinfo{title}{Elements of Distributed
  Algorithms: Modeling and Analysis with Petri Nets}}.
\newblock \bibinfo{publisher}{Springer Berlin Heidelberg},
  \bibinfo{address}{Berlin, Heidelberg}, \doi{10.1007/978-3-662-03687-7}.

\end{thebibliography}

\end{document}